\documentclass[final,5p,times,twocolumn]{elsarticletmp}
\journal{}

\usepackage{mathtools}
\mathtoolsset{showonlyrefs,showmanualtags}
 
\usepackage{datetime}
\ddmmyyyydate
\settimeformat{xxivtime}
\usepackage{amsmath}
\usepackage{amssymb}

\usepackage[abbrev]{bib_sty}

\usepackage{paralist}

\usepackage{graphicx}

\graphicspath{{./Fig/}}

\usepackage{datetime} 
\ddmmyyyydate
\settimeformat{xxivtime}

\newcommand{\tf[1]}{\ensuremath{\tilde{#1}}}
\newcommand{\crit}{\ensuremath{\mathcal{J}}} 
\newcommand{\df}{{\ensuremath{{\Delta f}}}}

\newcommand{\Z}{\ensuremath{\mathbb{Z}}}
\newcommand{\RLinf}{\ensuremath{\mathcal{RL}_\infty}}
\newcommand{\RHinf}{\ensuremath{\mathcal{RH}_\infty}}

\newcommand{\Fig}{Fig.}
\newcommand{\Sec}{Sec.}

\usepackage{amsthm}
\newtheorem{thm}{Theorem}
\newtheorem{lemma}[thm]{Lemma}

\newtheorem{assum}[thm]{Assumption}

\newtheorem{remark}[thm]{Remark}

\newcommand{\Colorone}{solid blue}%
\newcommand{\Colortwo}{dashed red}%
\newcommand{\Colorthree}{dash-dotted green}%
\newcommand{\Colorfour}{dotted magenta}%
\newcommand{\Colorfive}{solid cyan}%

\title{\LARGE \bf 
Sparse Iterative Learning Control with Application to a Wafer Stage: \\ Achieving Performance, Resource Efficiency, and Task Flexibility}

\begin{document}
\begin{frontmatter}
\title{%
Sparse Iterative Learning Control with Application to a Wafer Stage: \\ Achieving Performance, Resource Efficiency, and Task Flexibility%
}

\author[TUE]{Tom Oomen}
\author[KTH]{Cristian R. Rojas}
 \address[TUE]{Eindhoven University of Technology, Faculty of Mechanical Engineering, Control Systems Technology group, PO Box 513, 5600MB Eindhoven, The Netherlands.}
\address[KTH]{Automatic Control Lab, Electrical Engineering, KTH -- Royal Institute of Technology, S-100 44 Stockholm, Sweden.}

 \begin{abstract} 
Trial-varying disturbances are a key concern in Iterative Learning Control (ILC) and may lead to inefficient and expensive implementations and severe performance deterioration. The aim of this paper is to develop a general framework for optimization-based ILC that allows for enforcing additional structure, including sparsity. The proposed method enforces sparsity in a generalized setting through convex relaxations using $\ell_1$ norms. The proposed ILC framework is applied to the optimization of sampling sequences for resource efficient implementation, trial-varying disturbance attenuation, and basis function selection. The framework has a large potential in control applications such as mechatronics, as is confirmed through an application on a wafer stage. 

\end{abstract}
\end{frontmatter}
\section{Introduction} %

Iterative Learning Control (ILC) enables significant performance improvements for batch-to-batch control applications, by generating a command signal that compensates for repetitive disturbances through learning from previous iterations, also called batches or trials. Theoretical and implementation aspects, including convergence, causality, and robustness, have been addressed in, e.g., \cite{BristowThaAll2006}, \cite{AhnMooChe2007}, \cite{RogersGalOwe2007}, \cite{Owens2016}, \cite{PipeleersMoo2014}. Furthermore, successful applications have been reported in, e.g., robotics \cite{WallenDreGunRob2014}, mechatronics \cite{BolderZunKoeOom2017}, manufacturing \cite{HoelzleBar2016}, building control \cite{PengSunZhaTom2016}, nuclear fusion \cite{FeliciOom2015}, and rehabilitation \cite{FreemanHugBurChaLewRog2009}. However, several disadvantages of present ILC frameworks that limit further applications include 
\begin{inparaenum}[i)]
\item high implementation cost due to highly unstructured command signals, which are expensive to implement;\label{item:1}
\item amplification of trial-varying disturbances, including measurement noise;\label{item:2}
\item inflexibility to changing reference trajectories.\label{item:3}
\end{inparaenum}
The aim of the present paper is to develop an ILC framework that addresses these aspects \ref{item:1})-\ref{item:3}) by enforcing sparsity.

Regarding \ref{item:1}) ILC typically generates signals that require a large number of command signal updates thus leading to an expensive implementation. ILC directly learns from measured signals that are contaminated by trial-varying disturbances such as measurement noise. These trial-varying disturbances are often modeled as a realization of a stochastic process \cite{Ljung1999}. As a result, the ILC command signals have infinite support. In sharp contrast, command signals that are obtained through traditional feedforward designs, including \cite{LambrechtsBoeSte2005}, have finite support and are highly sparse. Command signals with a high number of non-zero elements, or another appropriate structural constraint, may lead to a prohibitively expensive implementation, e.g., in wireless sensor networks, wireless control applications, or embedded platforms with shared resources \cite{GoossensAzeChaDevGooKoeLiMirMolBeyNelSin2013}. Note that this is a different aspect than the actual computation of the command signal itself, which can be done in between subsequent tasks, see \cite{ZundertBolKoeOom2016b} for results in this direction.

Regarding \ref{item:2}), ILC typically amplifies trial-varying disturbances. In fact, typical ILC approaches amplify these disturbances by a factor of two, as is shown in the present paper. Approaches to attenuate trial-varying disturbances include norm-optimal ILC with appropriate input weighting \cite{BristowThaAll2006}, higher-order ILC for addressing disturbances with trial-domain dynamics \cite{GunnarssonNor2006}, and stochastic approximation-based ILC \cite{ButcherKar2011}. Also, a wavelet filtering-based approach is presented in \cite{MerryMolSte2008}, where a certain noise attenuation is achieved by setting certain wavelet coefficients to zero. In the present paper, a different approach is pursued to attenuate disturbances, where also wavelets immediately fit into the formulation, yet the sparsity can be enforced in an optimal way.

Regarding \ref{item:3}), changing reference signals typically lead to performance degradation of ILC algorithms \cite{BoerenBarKokOom2016}, since these essentially constitute trial-varying disturbances. This is in sharp contrast to traditional feedforward designs \cite{LambrechtsBoeSte2005} and is widely recognized in ILC designs.  A basis task approach is proposed in \cite{HoelzleAllWag2011}, where the command input is segmented. A basis function framework is developed and applied in \cite{WijdevenBos2010}, \cite{MeulenTouBos2008}, \cite{BolderOomKoeSte2014c} using polynomial basis functions, which is further extended to rational basis functions in \cite{ZundertBolOom2015}. These basis functions are typically selected based on prior information, e.g., based on the approach in \cite{LambrechtsBoeSte2005}, and trial-and-error. 

In model estimation and signal processing, the use of measured signals has comparable consequences, which has led to new regularization-based approaches that enforce sparsity. Early approaches include the non-negative garrote \cite{Breiman1995} and Least Absolute Shrinkage and Selection Operator (LASSO) \cite{Tibshirani1996}. These are further generalized in \cite{TibshiraniTay2011}, \cite{HastieTibWai2015}, \cite{BuhlmannGee2011}, \cite{BachJenMaiObo2011}. Related applications in system identification include \cite{RojasHja2011}, \cite{OhlssonLjuBoy2010}.

Although important developments have been made in ILC and several successful applications have been reported, present approaches do not yet exploit the potential of enforcing additional structure and sparsity. The aim of the present paper is to develop a unified optimization-based approach to ILC that allows for explicitly enforcing structure and sparsity, enabling improved resource efficiency, disturbance attenuation, and flexibility to varying reference signals. The approach employs convex relaxations, enabling the use of standard optimization routines. 

The main contribution of the present paper is a unified framework to sparse ILC. As subcontributions, trial-varying disturbances are analyzed in detail for explicit ILC algorithms (\Sec~\ref{sec:analyzeexplicitILC}). Subsequently, a general optimization-based framework to sparse ILC is developed (\Sec~\ref{sec:spilc}), including many specific cases that are relevant to ILC applications. The results are confirmed through an application to a wafer stage system (\Sec~\ref{sec:examples}). Related developments to the results in the present paper include the use of sparsity in control, where the main results have been related to Model Predictive Control (MPC), see \cite{AnnergrenHanWah2012}, \cite{KhoshfetratOhlLju2013}, \cite{Gallieri2015}. 

\emph{Notation:} Throughout, $\|x\|_{\ell_p}$ denotes the usual $\ell_p$ norm, $p \in \mathbb{Z}_{> 0}$. Also, $\|x\|_0 = \sum_{i} \mathbf{1} (x_i \neq 0)$, i.e., the cardinality of $x$. Note that $\|x\|_0$ is not a norm, since it does not satisfy the homogeneity property. It relates to the general $p$-norm by considering the limit $p\rightarrow 0$ of $\|x\|_p$. In addition, $\|\tf[X]\|_{\mathcal{L}_\infty}$ and $\|\tf[X]\|_{\mathcal{H}_\infty}$ denote the usual $\mathcal{L}_\infty$ and $\mathcal{H}_\infty$ norms of discrete time systems, respectively. Throughout, $J$ denotes a system that maps an input space to an output space, operating either over finite or infinite time, which follows from the context. In certain cases, the system is assumed linear, time invariant, and scalar, with transfer function representation $\tf[J]$. The spectrum of a signal $x$ is denoted $\phi_{x}$. 

\section{Problem formulation}\label{sec:probform}

\begin{figure}%
\centering
\includegraphics[width=.9\linewidth]{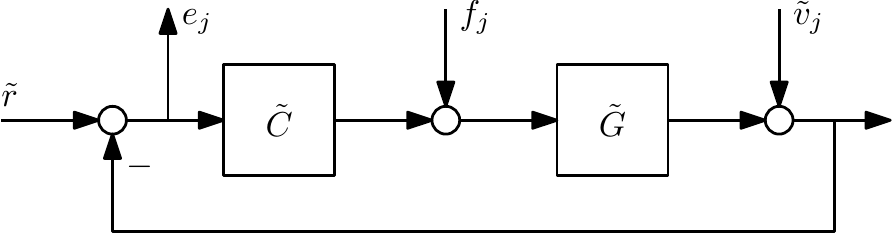}%
\caption{Parallel ILC structure \eqref{eq:parallelILC} as an example of \eqref{eq:generalILCsystem}.}
\label{fig:parallelILC}
\end{figure}

Consider the ILC system
\begin{equation}\label{eq:generalILCsystem}
e_{j} = r -J f_j -v_j
\end{equation}
be given, where $e_j\in \ell_2$ denotes the error signal to be minimized, $r\in \ell_2$ is the reference signal, $f_j \in \ell_2$ denotes the command signal, and $v_j \in \ell_2$ represents trial-varying disturbances, including measurement noise. Here and in the sequel, all signals are tacitly assumed to have appropriate dimensions. Furthermore, $J$ represents the true system, either open-loop or closed-loop, with causal and stable transfer function $\tf[J] \in \RHinf$. The index $j \in \Z_{\ge 0}$ refers to the trial number. Throughout, the command signal $f_{j+1}$ is generated by an ILC algorithm 
\begin{equation}\label{eq:generalILCupdate}
f_{j+1} = F(f_j, e_j),
\end{equation}
where the ILC update $F$ is defined in more detail later on. The general setup \eqref{eq:generalILCsystem} encompasses the parallel ILC setup in Figure~\ref{fig:parallelILC}, where 
\begin{equation}\label{eq:parallelILC}
e_j = S\tilde r - SGf_j - S \tilde v_j
\end{equation}
where $S$ follows from its transfer function $\tf[S] = \frac{1}{1+\tf[G]\tf[C]}$, $r = S\bar r$, $J = SG$, $v_j = S \tilde v_j$, and $\tf[C]$, $\tf[G]$ are assumed to be linear.
 
From \eqref{eq:generalILCupdate} and \eqref{eq:generalILCsystem}, it is immediate that the trial-varying disturbance $v_j$ directly affects the ILC command signal. In view of this observation, the problem investigated in this paper is to develop an ILC algorithm \eqref{eq:generalILCupdate} that satisfies the following  requirements:
\begin{compactenum}[{R}1)]
\item the iteration \eqref{eq:generalILCsystem}-\eqref{eq:generalILCupdate} is convergent over $j$;\label{item:converge}
\item the iteration \eqref{eq:generalILCsystem}-\eqref{eq:generalILCupdate} leads to a small error $e_j$ in the presence of trial-invariant disturbances $r$ and trial-variant disturbances $v_j$;\label{item:disturbanceattenuation}
\item the resulting command signal $f_j$ has a certain structure, including\label{item:resourceefficient}
\begin{compactenum}
\item \label{eq:fjspare} a small $\|f_j\|_0$, and/or,
\item \label{eq:dfjspare} a piecewise constant $f_j$ with a small number of jumps.
\end{compactenum}
\end{compactenum}
Here, R\ref{item:converge} is a basic requirement for any ILC algorithm and ensures stability in the trial domain, in addition to the assumed stability in the time domain that is guaranteed by stability of $J$ in \eqref{eq:generalILCsystem}, see also \cite{RogersGalOwe2007} for the stability of such two-dimensional systems. Requirement~R\ref{item:disturbanceattenuation} essentially states that the ILC algorithm should effectively compensate for $r$, while avoiding amplification of trial-varying disturbances $v_j$. Requirement~R\ref{item:resourceefficient} is imposed to enable resource-efficient implementations in terms of sampling or communication requirements, depending on the particular application requirements.

\section{Analysis of Trial-Varying Disturbances in Explicit ILC}\label{sec:analyzeexplicitILC}

In this section, trial-varying disturbances in ILC algorithms are analyzed. In particular, explicit linear ILC algorithms of the general form
\begin{equation}\label{eq:freqilcupdate}
f_{j+1} = Q (f_j + L e_j)
\end{equation}
are considered. The infinite time scalar case is considered, where $Q: \ell_2 \mapsto \ell_2$ and $L: \ell_2 \mapsto \ell_2$. Here, $Q$ and $L$ have associated transfer functions $\tf[Q] \in \RLinf$ and $\tf[L] \in \RLinf$. Note that $\tf[J] \in \RHinf$ reflects causality and stability of the system. The fact that $\tf[Q] \in \RLinf$ and $\tf[L] \in \RLinf$ reflects that typical ILC algorithms are typically non-causal, and are usually implemented such that bounded solutions are obtained through finite-time preview or via stable inversion through a bilateral $Z$-transform \cite{ZundertBolKoeOom2016b}. %

The trial-varying disturbance $v_j$ in \eqref{eq:generalILCsystem} will propagate throughout the iterations through the iteration-domain update \eqref{eq:freqilcupdate}. The following assumption is widely adopted \cite{Ljung1999}.
\begin{assum}\label{assum:noise}
Let $v_j = H n_j$, where $n_j$ is i.i.d. zero-mean white noise with variance $\lambda_e$, $\tf[H]$ monic and bistable.
\end{assum}
Clearly, $v_j$ typically does not have compact support. As a result, $f_{j+1}$ will not have compact support in general due to ILC algorithm \eqref{eq:freqilcupdate}.

To enable a more detailed analysis, the following auxiliary result provides a suitable condition to guarantee that the iteration defined by \eqref{eq:generalILCsystem} and \eqref{eq:freqilcupdate} converges.

\begin{thm}\label{thm:contractionmap}
The iteration defined by \eqref{eq:generalILCsystem} - \eqref{eq:freqilcupdate} converges monotonically in the $\ell_2$ norm to a fixed point $f_\infty$ and resulting $e_\infty$ iff
\begin{equation}
\|\tf[Q](1-\tf[L]\tf[J])\|_{\mathcal{L}_\infty} < 1.
\end{equation}
\end{thm}
\begin{proof}
Substituting \eqref{eq:generalILCsystem} into \eqref{eq:freqilcupdate} leads to 
\begin{math}
f_{j+1} = Q(I-LJ)f_j + QLr - QLv_j.
\end{math}
Using transfer function representations and subsequent application of the Banach fixed-point theorem in conjunction with  \cite[Theorem 4.4]{ZhouDoyGlo1996}  yields the desired result. 
\end{proof}
Note that Theorem~\ref{thm:contractionmap} allows for non-causal ILC algorithms, i.e., $Q, L \in \RLinf$. This is more general compared to related analyses, including \cite[Chapter 3]{Moore1993}, which only allow for causal ILC algorithms by restricting to the $\mathcal{H}_\infty$ norm. 

The following result is the main result of this section and reveals the propagation of noise in the iteration defined by \eqref{eq:generalILCsystem} and \eqref{eq:freqilcupdate}.
\begin{thm}\label{thm:noiseanalysis}
Given the system \eqref{eq:generalILCsystem} and ILC update \eqref{eq:freqilcupdate} with $f_0 = 0$,  Assumption~\ref{assum:noise}, and that the iteration is stable in the sense of Theorem~\ref{thm:contractionmap}, then, 
\begin{equation}\label{eq:limiterrorspectrum}
\textstyle
\phi_{e_\infty} =
\left| 
\frac{1- \tf[Q]}{1-\tf[Q](1-\tf[L]\tf[J])}
\right|^2 \phi_r 
+
\left(
1
+
\frac{\left|
\tf[J]\tf[Q]\tf[L]
\right|^2}{1-|\tf[Q](1-\tf[L]\tf[J])|^2}
\right)\phi_v.
\end{equation}
\end{thm}

Theorem~\ref{thm:noiseanalysis} provides a detailed analysis of the propagation of noise for the general ILC algorithm \eqref{eq:freqilcupdate}. In special cases, the result can be further simplified. For instance, in inverse-model ILC, $\tf[Q]= 1$ and $\tf[L] = \tf[J]^{-1} \in \RHinf$, in which case Theorem~\ref{thm:noiseanalysis}  reveals that
\begin{equation}\label{eq:noiseamplification}
\phi_{e_\infty} = 2 \phi_v.
\end{equation}
The result \eqref{eq:noiseamplification} reveals that the limit error spectrum involves an amplification of the noise spectrum by a factor of two. 

Inclusion of a learning gain $\alpha \in (0,1]$ in inverse-model ILC, i.e., replacing \eqref{eq:freqilcupdate}
 by $f_{j+1} = Q (f_j + \alpha L e_j)$, mitigates the amplification of trial-varying disturbances, i.e., 
\begin{equation}
\phi_{e_\infty} = \left(
1 + \frac{\alpha^2}{2\alpha - \alpha^2}
\right)\phi_v.
\end{equation}
By taking $\alpha \rightarrow 0$, a  first-order Taylor series approximation yields
\begin{equation}\label{eq:rolealpha}
\phi_{e_\infty} \approx \left(
1 + \frac{1}{2} \alpha
\right)\phi_v.
\end{equation}
Hence, choosing $\alpha$ small leads to a limit error $\phi_{e_\infty} = \phi_v$, which intuitively corresponds to the optimal result, since the iteration-domain feedback \eqref{eq:freqilcupdate} cannot attenuate $v_j$ in iteration $j$. An alternative to attenuate $\phi_v$ is to re-design the controller $C$ in \eqref{eq:parallelILC}, which should from a disturbance attenuation perspective be designed such that $\tf[S] \approx \tf[H]^{-1}$, as is advocated in \cite{BoerenBruOom2017}. Note that this affects $J$ in \eqref{eq:generalILCsystem}.%

\begin{remark}
The results in this section rely on infinite time signals and LTI systems. Alternative ILC designs based on finite-time optimization \cite{BristowThaAll2006}, see also the forthcoming section, explicitly address the boundary effects, typically leading to an LTV ILC update \eqref{eq:generalILCupdate}, even if $J$ is LTI. In~\cite{ZundertBolKoeOom2016b}, it is shown that these optimization-based designs are equivalent to a certain linear-quadratic-tracking problem. As a result, the solution reaches a certain stationary value for sufficiently long task lengths, in which case an LTI $L$ and $Q$ can be derived for which the results of Theorem~\ref{thm:noiseanalysis} apply. This also implies that the design of weighting filters for such optimization-based design can be further investigated, as is briefly summarized in the next section.%
\end{remark}

\section{Sparse ILC}\label{sec:spilc}

In this section, the general optimization-based ILC framework is presented that allows for enforcing additional structure compared to alternative ILC structure. In fact, traditional norm-optimal ILC algorithms \cite{BristowThaAll2006} are recovered as a special case. In the next subsection, the general framework is presented and motivated, followed by specific design choices in the subsequent sections.

\subsection{General approach}\label{sec:generalapproach}

Throughout, the criterion
\begin{equation}\label{eq:gencrit}
\begin{split}
\crit(f_{j+1}) = &
\frac{1}{2}
\|W_e e_{j+1}
\|_2^2
+
\frac{1}{2}
\|W_f f_{j+1}
\|_2^2
\\&+
\frac{1}{2}
\|W_{\df} 
\left(
f_{j+1} - f_j
\right)
\|_2^2
+ \lambda \|D f_{j+1} \|_1
\end{split}
\end{equation}
is considered. Here, finite time signals of length $N$ are considered to obtain an optimization problem with a finite number of decision variables, i.e., $e_j, f_j \in \mathbb{R}^N$. The matrices are defined in the sequel and are assumed to have compatible dimensions. In addition, existence of a unique solution is typically assumed, which can be directly enforced by assuming appropriate positive (semi-) definiteness assumptions on the design variables $W_e$, $W_f$, $W_\df$, $D$, and $\lambda$. Also, $e_{j+1}$ in \eqref{eq:gencrit} is considered to be the noise-free prediction $e_{j+1} = r-Jf_{j+1}$. Since also $r$ is unknown, the main idea in ILC is to use this approximation also for $e_j$, leading to 
\begin{equation}\label{eq:iterativeej}
e_{j+1} = e_j - J(f_{j+1} - f_j),
\end{equation} 
where $e_j$ is the measured error signal during trial $j$. Thus, substituting \eqref{eq:iterativeej} into \eqref{eq:gencrit} renders the optimization problem as a function of the known variables $e_j, f_j$, user-defined variables, and the decision variable $f_{j+1}$.

The motivation for considering \eqref{eq:gencrit} is as follows. First, if $\lambda = 0$, then standard norm-optimal ILC is recovered, e.g., as in \cite{GunnarssonNor2001}. In this case, an analytic solution of the form \eqref{eq:freqilcupdate} is directly obtained with
\begin{align}
L &= (J^T  \bar W_e J + W_{\df})^{-1}J^T \bar W_e \label{eq:NOILC1} \\
Q &= (J^T \bar W_e J + \bar W_f + \bar W_{\df})^{-1}(J^T \bar W_e J + \bar W_{\df}),\label{eq:NOILC2}
\end{align}
where $\bar W_e = W_e^T W_e$, $\bar W_f = W_f^TW_f$, and $\bar W_\df = W_\df^TW_\df$.

The second motivation stems from the observation that the terms $\frac{1}{2}
\|W_f f_{j+1}
\|_2^2$
and
$\frac{1}{2}
\|W_{\df} 
\left(
f_{j+1} - f_j
\right)
\|_2^2$
essentially involve a ridge regression or Tikhonov regularization. If $f_j = 0$, then the two terms coincide. If $f_j \neq 0$, i.e., during the ILC iterations, then $W_f$ typically leads to $Q \neq I$ in \eqref{eq:freqilcupdate}, providing robustness with respect to modeling errors \cite{Bristow2008}. Increasing $W_\df$ attenuates trial-varying disturbances, which is similar to reducing $\alpha$ in \eqref{eq:rolealpha}. Note that $W_f$ also plays a small role to decrease trial-varying disturbances, since it essentially leads to a smaller mean-square error. However, it leads to a non-zero limit error $e_\infty$, even in the absence of $v_j$ due to the weight on $f_j$, which coincides with a $\tf[Q] \neq 1$ in Theorem~\ref{thm:noiseanalysis}. 

The third and main motivation for considering the extended criterion \eqref{eq:gencrit} is the additional term $\lambda \|D f_j \|_1$ that is used to enforce sparsity and structure. Note that sparsity is measured directly through the $\ell_0$ norm. However, inclusion of an $\ell_0$ penalty in the criterion  \eqref{eq:gencrit}  leads to a non-convex optimization problem, which in fact is NP-hard, see \cite{Natarajan1995}. The $\ell_1$ norm is a convex relaxation of the $\ell_0$ norm. To see this, note that \eqref{eq:gencrit} is essentially in Lagrangian form. For the purpose of explanation, consider the simplified form by selecting $W_e = I$, $j = 1$, $f_0 = 0$, $W_f = 0$, $W_\df = 0$, and $D = I$. Using \eqref{eq:iterativeej}
\begin{equation}\label{eq:l1simplified}
\crit(f_1) = \frac{1}{2} \|e_0 - J f_1 \|_2^2 + \lambda \|f_1\|_1,
\end{equation}
which is equivalent to the primal optimization problem
\begin{equation}
\begin{aligned}
& \min_{f_1}
& & \frac{1}{2}\|e_0 - J f_1 \|_2^2  \\
& \text{subject to}
& & \|f_1\|_1 \leq t. \label{eq:l1primal}
\end{aligned}
\end{equation}
for the range of $t$ where the constraint in \eqref{eq:l1primal} is active. This implies that for a given value of $\lambda$, there exists a value of $t$ for which \eqref{eq:l1simplified} and \eqref{eq:l1primal} have identical minima.  In this simplified case, the interpretation in \cite{Tibshirani1996} applies to the ILC problem. In particular, the constraint in \eqref{eq:l1primal} is plotted  in \Fig~\ref{fig:motivationellone} in addition to several elliptical contour lines of the objective function in \eqref{eq:l1primal}. The solution to \eqref{eq:l1primal} corresponds to the smallest ellipsoid that touches the rhombus of the constraint. If this happens at the corner, as is common and also in this case, then one of the coefficients is zero and a sparse solution is obtained. 

In contrast, traditional norm-optimal ILC, i.e., corresponding to the solution \eqref{eq:NOILC1} - \eqref{eq:NOILC2}, typically does not lead to a sparse solution with zero entries in $f_1$. To see this, consider a similar simplified case as in \eqref{eq:l1simplified}
\begin{equation}\label{eq:l2simplified}
J(f_1) = \frac{1}{2} \|e_0- J f_1 \|_2^2 + \tau \|f_1\|_2,
\end{equation}
which is again in Lagrangian form. Here, $\tau$ directly relates to the weights in \eqref{eq:gencrit} if $W_f$ and $W_\df$ are selected as the common diagonal case with initialization $f_0 = 0$. The primal optimization problem corresponding to \eqref{eq:l2simplified} is given by
\begin{equation}
\begin{aligned}
& \min_{f_1}
& & \frac{1}{2}\|e_0 - J f_1 \|_2^2  \\
& \text{subject to}
& & \|f_1\|_2 \leq t. \label{eq:l2primal}
\end{aligned}
\end{equation}
In \Fig~\ref{fig:motivationellone}, the constraint is again shown together with the contour lines of the objective function. Due to the lack of corners of the constraint, the presence of zeros in the solution of \eqref{eq:l2primal} is very unlikely in general. Hence, the $\ell_1$ norm promotes sparse solutions, whereas the $\ell_2$ norm in general does not. 

\begin{figure}%
\centering
\includegraphics[width=.75\linewidth]{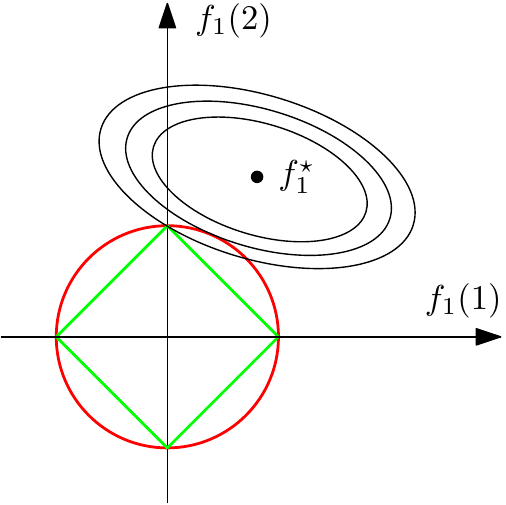}%
\caption{Enforcing sparsity in ILC. Assuming $N = 2$, $f_1$ contains two elements. The constraint set, i.e., the $\ell_1$ ball is plotted in green. In addition, ellipsoidal contour lines corresponding to the objective in \eqref{eq:l1primal} are plotted. The optimal solution is found at the point where the contour line first touches the constraint set, which in this case implies $f_1(1) = 0$, hence $f_1$ is sparse. In contrast, in the ridge regression case of \eqref{eq:l2primal} (whose constraint is shown in red), the solution is not sparse. In particular, this solution is obtained when the contour lines of the objective function in \eqref{eq:l2primal} first touches the constraint set corresponding to the $\ell_2$ ball. In addition, $f_1^{\star}$ denotes the unconstrained solution to the objective function in \eqref{eq:l1primal} and \eqref{eq:l2primal}.}
\label{fig:motivationellone}
\end{figure}

Finally, it is remarked that if $\lambda > 0$, then the solution to \eqref{eq:gencrit} typically cannot be obtained in closed-form as in \eqref{eq:NOILC1}-\eqref{eq:NOILC2}. Interestingly, a unique solution to \eqref{eq:gencrit} exists due to convexity. The optimization problem \eqref{eq:gencrit} can be readily solved using general convex optimizers. In addition, several efficient algorithms have been developed, see, e.g., \cite[Chapter 5]{HastieTibWai2015} for an overview. Several of such algorithms provide the entire solution path as a function of $\lambda$. The particular algorithm depends on the choice of $D$, but several relevant choices are outlined below.

\subsection{Sparse command signals via lasso}\label{sec:lassoilc}

In view of requirement R\ref{eq:fjspare} in \Sec~\ref{sec:probform}, in certain applications it is required to have a sparse command signal $f_j$. To this end, $D$ in \eqref{eq:gencrit} can be selected as $D = I$. As a result, the value of $\lambda > 0$ will dictate the sparsity of the solution. In addition, in this classical lasso approach, $W_f$ and $W_\df$ may be selected as $W_f = 0$ and $W_\df = 0$, i.e., traditional design guidelines for norm-optimal ILC regarding positive definiteness of these matrices, as in \cite{GunnarssonNor2001}, need not be considered, even for the situation where $J$ is singular. The resulting criterion becomes
\begin{equation}\label{eq:lassocrit}
\begin{split}
\crit(f_{j+1}) = &
\frac{1}{2}
\|W_e (e_j - J f_{j+1})
\|_2^2
+ \lambda \|f_{j+1}\|_1,
\end{split}
\end{equation}
which closely reflects the original lasso approach in \cite{Tibshirani1996}.

\subsection{Elastic net lasso}\label{sec:elasticnet}
In the lasso ILC approach in \Sec~\ref{sec:lassoilc}, the commonly used weighting matrices in $W_f$ and $W_\df$ are set to zero. Interestingly, by selecting either $W_f$ or $W_{\df}$ unequal to zero, an ILC algorithm that relates to the elastic net is obtained, see \cite{ZouHas2005}, which combines lasso and ridge regression. An important advantage is that the elastic net improves group sparsity, where several components become zero simultaneously. Notice that a drawback of the so-called naive elastic net, which coincides with $W_f \neq 0$, $W_{\df} = 0$, is that it leads to a double shrinkage, and it benefits from a correction step \cite{ZouHas2005}. In contrast, in ILC the alternative choice $W_f = 0$, $W_{\df} \neq 0$ can be made, which enforces sparsity in addition to attenuating trial-varying disturbances, see \Sec~\ref{sec:generalapproach}.

\subsection{Sparse updates via fused lasso}\label{sec:fusedlasso}

In view of Requirement R\ref{eq:dfjspare}, it may be required that the signal $f_j$ is not necessarily sparse but piecewise constant, i.e., its value is only changed occasionally in time. This requires a certain structure of signal, which is different than sparsity $\|f_j\|_0$. The main idea is to select $D$ as 
\begin{equation}\label{eq:fusedlassoD}
D_{f} = 
\begin{bmatrix}
-1 & 1 & \\
& -1 & 1 \\
& & \ddots & \ddots\\
& & & -1 & 1
\end{bmatrix},
\end{equation}
a choice which is also known as the fused lasso, e.g., \cite{TibshiraniSauRosZhuKni2005} and leads to the criterion
\begin{equation}\label{eq:fusedlassocrit}
\begin{split}
\crit(f_{j+1}) = &
\frac{1}{2}
\|W_e (e_j - J f_{j+1})
\|_2^2
+ \lambda \|D_ff_{j+1} \|_1.
\end{split}
\end{equation}
Interestingly, the fused lasso \eqref{eq:fusedlassocrit} can be recast as a traditional lasso of the form \eqref{eq:lassocrit}, yet with an increment-input-output system description. To establish the connection, let $\tf[J^i] = \tf[J] (1-z^{-1})$ be the increment-input-output system. Also, expand $D_f$ in \eqref{eq:fusedlassoD} as 
\begin{equation}\label{eq:fusedlassoD2}
D_{f}^i = 
\begin{bmatrix}
1 \\
-1 & 1 & \\
& -1 & 1 \\
& & \ddots & \ddots\\
& & & -1 & 1
\end{bmatrix}.
\end{equation}
Then, a change of variables
\begin{equation}
f^i_{j+1} = D_f^i f_{j+1},
\end{equation}
where $f^i_{j+1} $ denotes the incremental input,
leads to
\begin{equation}\label{eq:fusedlassocrittransformed}
\begin{split}
\crit(f_{j+1}) = &
\frac{1}{2}
\|W_e (e_j - J^i f_{j+1}^i)
\|_2^2
+ \lambda \|f_{j+1}^i \|_1,
\end{split}
\end{equation}
with $J^i = J (D_f^i)^{-1}$ corresponding to $\tf[J^i]$.

\subsection{Sparse fused lasso}\label{sec:sparsefusedlasso}
Up to this point, Requirement R\ref{eq:fjspare} and Requirement R\ref{eq:dfjspare} have been addressed separately in \Sec~\ref{sec:lassoilc} and \Sec~\ref{sec:fusedlasso}, respectively. In certain applications, it may be desired to impose both Requirement R\ref{eq:fjspare} and Requirement R\ref{eq:dfjspare}. 

Interestingly, Requirement R\ref{eq:fjspare} and Requirement R\ref{eq:dfjspare} can be enforced both by selecting
\begin{align}\label{eq:sparsefusedlasso}
D = \begin{bmatrix}
\alpha D_f & I
\end{bmatrix},
\end{align}
in \eqref{eq:gencrit}. Here, the parameter $\lambda$ can still be chosen to enforce sparsity, i.e., Requirement R\ref{eq:fjspare}, whereas the additional tuning parameter $\alpha \in \mathbb{R}_{\geq 0}$ enforces Requirement R\ref{eq:dfjspare}. This leads to the so-called sparse fused lasso \cite{TibshiraniTay2011}. Note that additional requirements can easily be incorporated using a similar construction as \eqref{eq:sparsefusedlasso}.

\subsection{Basis function ILC}

In recent extensions to ILC, several basis functions are employed. On the one hand, wavelet basis functions are used in, e.g., \cite{MerryMolSte2008}. These immediately fit in the formulation \eqref{eq:gencrit}, see also \cite[\Sec\ 2.1.3]{TibshiraniTay2011}, enabling a systematic way for thresholding while explicitly addressing the performance criterion. 

On the other hand, flexibility to varying reference signals is achieved by employing basis functions that depend on the reference. In particular, the command signal is parameterized as $f_{j+1} = \Psi(r) \theta_{j+1}$, see, e.g.,  \cite{WijdevenBos2010}, \cite{MeulenTouBos2008}, \cite{BolderOomKoeSte2014c}, \cite{ZundertBolOom2015}. The proposed framework can be employed to minimize the number of required basis functions. For instance, a large set can be postulated, e.g., following the guidelines in \cite{LambrechtsBoeSte2005}. Next, an alternative formulation of \eqref{eq:gencrit} can be considered, e.g.,
\begin{equation}\label{eq:gencritbasis}
\begin{aligned}
& \min_{\theta_{j+1}}
& & \|\theta_{j+1}\|_1  \\
& \text{subject to}
& & \frac{1}{2}
\|W_e e_{j+1}
\|_2^2
+
\frac{1}{2}
\|W_f \Psi(r) \theta_{j+1}
\|_2^2
\\ & & &\quad +
\frac{1}{2}
\|W_{\df} \Psi(r) 
\left(
\theta_{j+1} - \theta_{j}
\right)
\|_2^2 \leq t,
\end{aligned}
\end{equation}
where a suitable value of $t$ can be obtained by solving the standard norm-optimal ILC in \eqref{eq:NOILC1}-\eqref{eq:NOILC2}.

\subsection{Extensions, analysis, and discussion} \label{sec:extensions}
A general framework for enforcing sparsity and structure in iterative learning control has been proposed, and several specific choices have been outlined. Further extensions that are beyond the scope of the present paper but can be directly incorporated include group lasso \cite{YuanLin2006}, adaptive lasso \cite{Zou2006}, reweighted $\ell_1$ \cite{CandesWakBoy2008}, and the use of non-convex penalties \cite{BertsimasKinMaz2016}. 

\subsubsection{Reestimation for debiasing}
Note that the lasso shrinks the estimate compared to the least-squares terms in \eqref{eq:gencrit}. Through a reestimation step of the nonzero coefficients, debiasing is obtained. Note that in certain cases, the bias helps to obtain a smaller overall error, i.e., including both bias and variance aspects, which closely relates to the well-known Stein estimator \cite{Stein1956}. However, for ILC such a bias is undesired, since it is automatically eliminated by performing iterations, see Theorem~\ref{thm:contractionmap}. Thus, it is expected that as the ILC iterations increase, the advantages of reestimating for debiasing become more important. Similar reestimation steps are proposed in \cite{RojasHja2011}, \cite{RojasTotHja2014}, \cite[Page 439]{Murphy2012}, \cite[\Sec \ 7.1]{KimKohBoyGor2009}. Interestingly, in the context of ILC, the idea of enforcing sparsity followed by a reestimation step essentially has the same role as a $Q$-filter in traditional ILC, see \cite{BoerenBarKokOom2016} for details. 

\subsubsection{Sparse signal recovery}
The main motivation for using the $\ell_1$ norm in \eqref{eq:gencrit} essentially is to provide a convex relaxation of the $\ell_0$ norm. In case the optimal command input, i.e., for $j \rightarrow \infty$ and $v_j = 0$, the signal $f_{\infty}$ that minimizes $J(f_\infty)$, is sparse, a relevant question is whether this optimal sparse vector can be recovered using the formulation \eqref{eq:gencrit}. The answer depends on the sparsity of the underlying optimal command input $f_j$, as well as on the matrix $J$. In \cite{CandesTao2005}, a sufficient condition that relies on the restriced isometry property is provided. However, these conditions are violated for many practical cases. Nonetheless, the formulation \eqref{eq:gencrit} provides an effective way to enforce sparsity. 

\subsubsection{Monotonic convergence}
Monotonic convergence is a commonly used requirement for practical applications. Indeed, it is well-known that poorly designed ILC algorithms can lead to a significant learning transient. It is well-known that traditional norm-optimal ILC, i.e., setting $\lambda = 0$ in \eqref{eq:gencrit}, is monotonically convergent in $f_j$, see, e.g., \cite{Bristow2008}, where the usual assumption $v_j = 0$ is tacitly assumed to analyze monotonic convergence. However, if $\lambda > 0$, the criterion \eqref{eq:gencrit} involves multiple norms, i.e., both the $\ell_1$ and the $\ell_2$ norm. As a result, monotonic convergence requires a more detailed analyis. 

To proceed, consider for instance the elastic net lasso of \Sec~\ref{sec:elasticnet} with $D = I$, $W_f = 0$, $W_\df \succ 0$. In this case, monotonic convergence of the ILC cannot be guaranteed in general if $\lambda = 0$. Interestingly, in this case the criterion \eqref{eq:gencrit} can be recast as
\begin{equation}\label{eq:gencritmonconv}
\begin{split}
\crit(f_{j+1}) = &
\frac{1}{2}
\left\|
\left(
\begin{bmatrix}
W_e e_j \\ 0
\end{bmatrix}
+
\begin{bmatrix}
W_e \\ W_\df
\end{bmatrix}f_j
\right)
-
\begin{bmatrix}
W_e \\ W_\df
\end{bmatrix}
f_{j+1}
\right\|_2^2
\\ &+ \lambda \| f_{j+1} \|_1.
\end{split}
\end{equation}
Next, there exists a value of $\tau$ such that the optimization problem
\begin{equation}\label{eq:critformonconv}
\begin{aligned}
& \min_{f_{j+1}}
& & \| f_{j+1} \|_1  \\
& \text{subject to}
& &\frac{1}{2}
\left\|
\begin{bmatrix}
W_e e_j \\ 0
\end{bmatrix}
-
\begin{bmatrix}
W_e \\ W_\df
\end{bmatrix}
(f_{j+1}-f_j)
\right\|_2^2 \leq \tau.
\end{aligned}
\end{equation}
has an identical solution as \eqref{eq:gencritmonconv} at a certain iteration $j$. If $\tau$ is fixed, then the criterion \eqref{eq:critformonconv} can be directly used to enforce monotonic convergence of $f_j$ in the $\ell_1$-norm.

\section{Application to a Wafer Stage}\label{sec:examples}

\subsection{Setup}

\begin{figure}%
\centering
\fbox{\includegraphics[width=.9\linewidth]{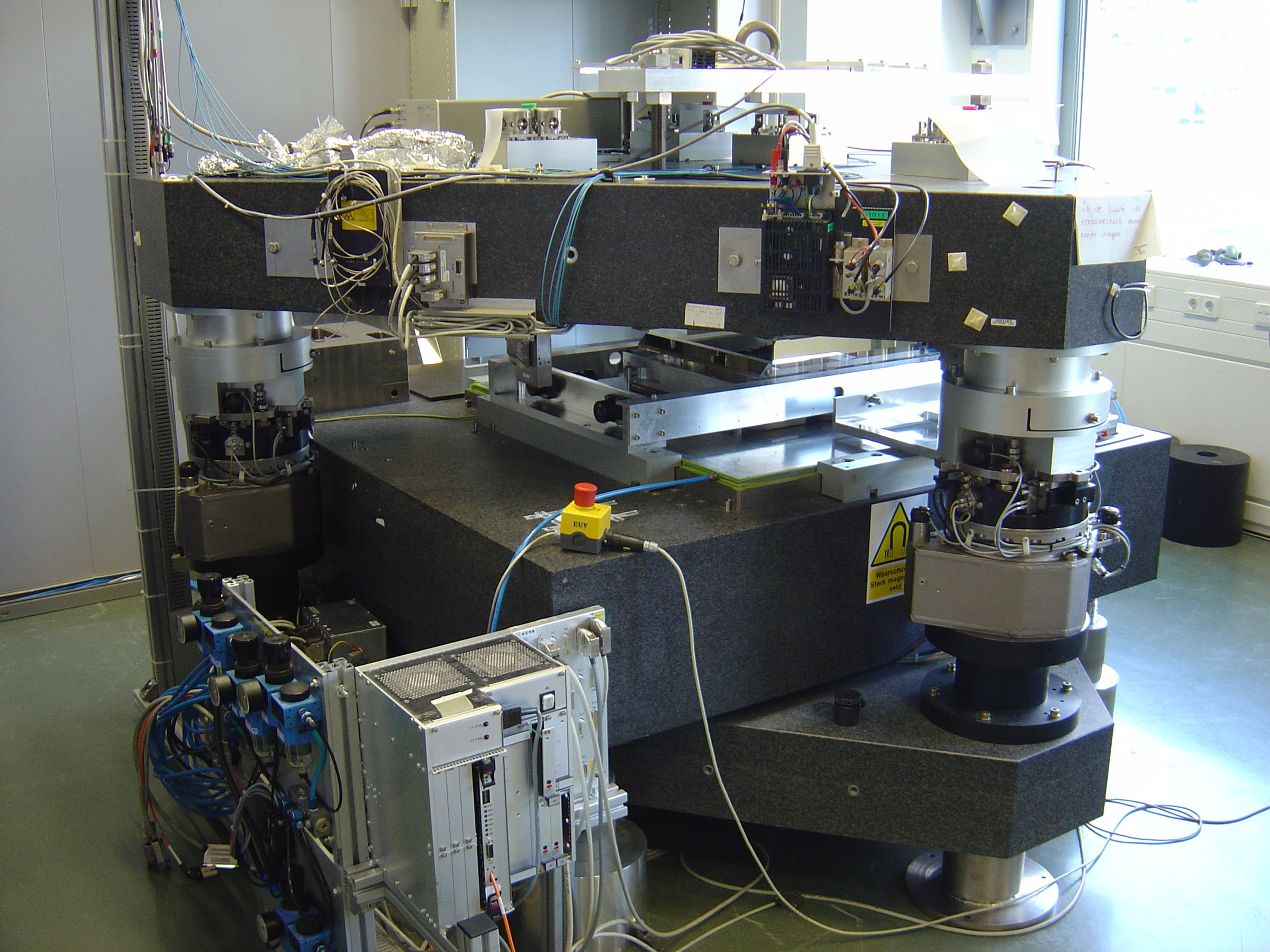}}%
\caption{Considered wafer stage application.}
\label{fig:systemdef}
\end{figure}

The considered system is a wafer stage, see  \Fig~\ref{fig:systemdef}. Wafer stages are positioning systems that are used in the production of integrated circuits (ICs) through a photolithographic process. The considered wafer stage is controlled in all six motion degrees-of-freedom, i.e., three translations and three rotations. The system is a dual-stage system, where the long stroke enables a stroke of $1 \ \mathrm{m}$ in the horizontal plane, whereas the short stroke enables a positioning accuracy of $1 \ \mathrm{nm}$. Further details on the system and the considered actuation and sensor system are provided in \cite{OomenHerQuiWalBosSte2014}. Throughout, a sampling frequency of $1 \ \mathrm{kHz}$ is adopted, as in \cite{Oomen2014}.

To enable a detailed comparison between the various approaches in \Sec~\ref{sec:spilc}, the identified model in \cite{Oomen2014} is considered as true system, i.e., the result as described in \cite{Oomen2014} is denoted $G_o$. In addition, the feedback controller designed in \cite{Oomen2014} is adopted to stabilize the system. In \Fig~\ref{fig:system}, the open-loop $G_o$ and closed-loop $S_o G_o$  are depicted. In addition, a closed-loop model is made, where a model error is introduced by selecting $J = 0.7 S_oG_o$. This model error is introduced to investigate robust convergence properties of ILC. The resulting model $J$ is also depicted in \Fig~\ref{fig:system}.

The additive noise $\tilde v_j$ is zero mean white noise with a normal distribution and variance $\lambda_e = 1.5 \cdot 10^{-7}$. As a result, $H$ in Assumption~\ref{assum:noise} has transfer function $\tilde H = \frac{1}{1+\tilde G \tilde C}$.

\begin{figure}%
\centering
\includegraphics[width=.9\linewidth]{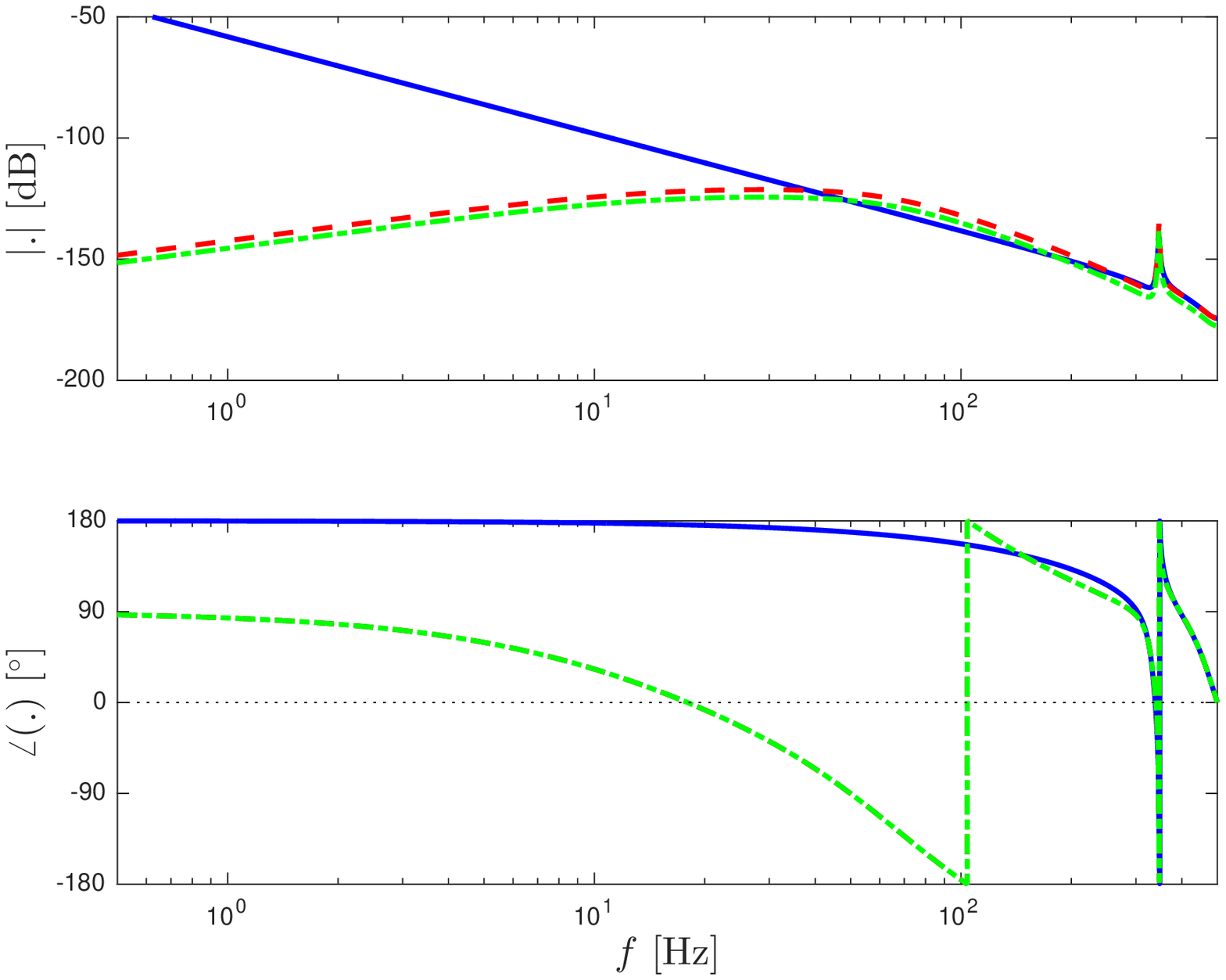}%
\caption{Open-loop true system $G_o$ in (\Colorone), closed-loop true system $S_oG_o$ (\Colortwo), closed-loop model $J$ (\Colorthree).}
\label{fig:system}
\end{figure}

\begin{figure}%
\centering
\includegraphics[width=.9\linewidth]{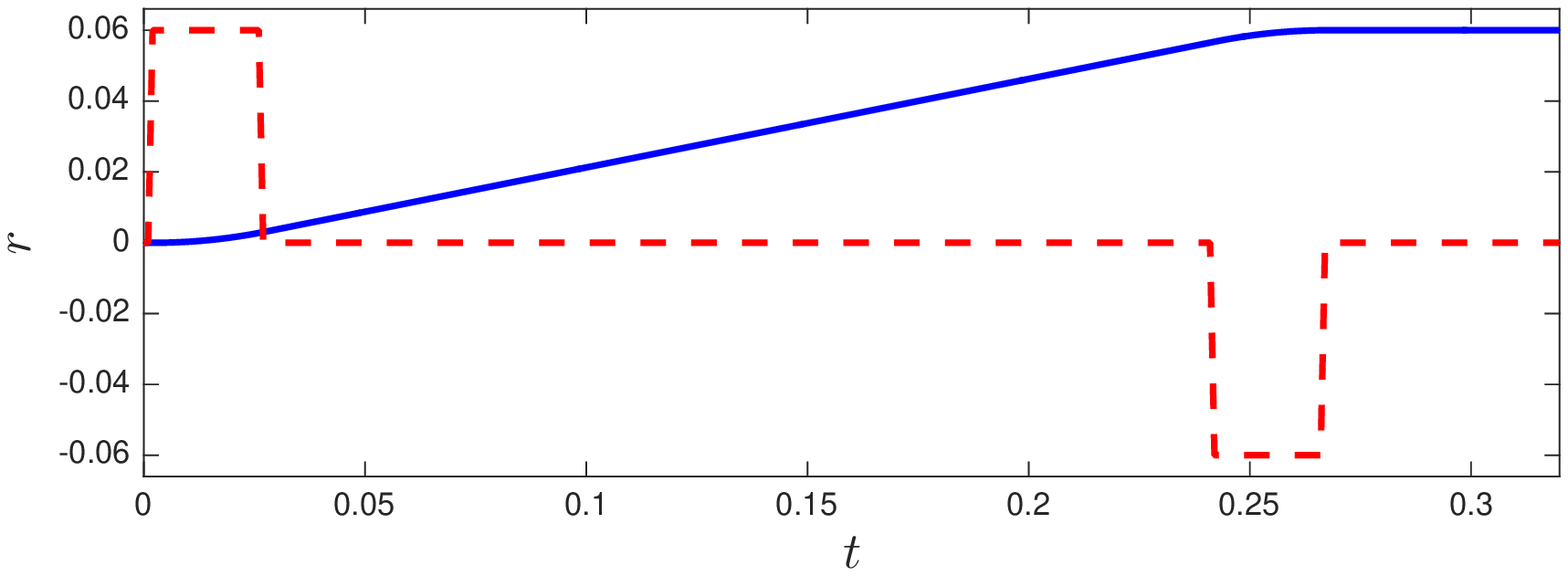}%
\caption{Reference $r$ in \eqref{eq:generalILCsystem} (\Colorone), scaled acceleration profile (\Colortwo).}
\label{fig:task}
\end{figure}

The task $r$ is shown in \Fig~\ref{fig:task}, which is a position signal. In addition, the corresponding scaled acceleration profile is depicted, which is expected to constitute the main contribution of $f_j$ \cite{LambrechtsBoeSte2005}, \cite{MeulenTouBos2008}. For the considered wafer stage application in \Fig~\ref{fig:systemdef}, the constant velocity phase is most important for performance, see \cite[\Fig\ 16 and \Fig\ 20]{Butler2011}, which takes place between $0.03\ \mathrm{s}$ and $0.24 \ \mathrm{s}$.

The goal of this section is to illustrate and compare the proposed approaches in Sec~\ref{sec:spilc}. The reference situation, i.e., feedback only with $f_0 = 0$ in \Fig~\ref{fig:parallelILC} is shown in \Fig~\ref{fig:ilctikresults} (\Colorone), \Fig~\ref{fig:ilctikresults2}, and \Fig~\ref{fig:spectra}. In particular, the approaches in \Sec~\ref{sec:spilc} are applied in this section. %

\subsection{Traditional Norm-Optimal ILC}\label{sec:exampleNOILC}

First, the traditional norm-optimal ILC solution is implemented with $\lambda = 0$ in \eqref{eq:gencrit} with the analytic solution \eqref{eq:NOILC1}-\eqref{eq:NOILC2}. Here, $W_e = I$, $W_f = 0$, and $W_{\df} = 10^{-10}$. Notice that $W_{\df}$ is relatively small but nonzero, since a nonzero $W_{\df}$ or $W_f$ is required to enforce a unique optimal solution. 

The results after $40$ iterations are depicted in \Fig~\ref{fig:ilctikresults}. Clearly, the error is reduced to a very small value. As is expected, the feedforward is nonzero at every time instant and very noisy. 

To further analyze these results, the $2$-norm of the stochastic, i.e., trial-varying, part of the error is computed as $\sqrt{\sum_{t = 1}^N (e_j(t) - \hat e_\infty(t))^2}$, see \Fig~\ref{fig:ilctikresults2}. Here, $\hat e_\infty $ is computed as follows. After a sufficient number of iterations $n_{\text{conv}}$, the ILC algorithm is assumed to have converged, after which an additional iterations $n_{\text{iter}}$ is used to compute
\begin{math}
\hat e_\infty = \frac{1}{n_{\text{iter}}} \sum_{j = n_{\text{conv}}}^{n_{\text{conv}}+n_{\text{iter}}-1}e_j.
\end{math}
Clearly, \Fig~\ref{fig:ilctikresults2} reveals that the trial-varying part of the error is amplified by a factor $2$, which corroborates the result of Theorem~\ref{thm:noiseanalysis}, where $Q \approx 1$ due to the specific selection of weighting filters. 

To further investigate the amplification of trial-varying disturbances, the spectrum of the trial-varying part of the errors in \Fig~\ref{fig:ilctikresults2}  is estimated, see \Fig~\ref{fig:spectra}. In addition, the spectrum $\phi_v = \left| \frac{1}{1+\tilde G \tilde C}\right|^2 \lambda_e$ is computed, as well as $2\phi_v$. Again, this clearly confirms the result of Theorem~\ref{thm:noiseanalysis}. In particular, the presented ILC approach with $\lambda = 0$ and small $W_f$ and $W_\df$ leads to a perfect attenuation of trial-invariant disturbances. However, it amplifies trial-varying disturbances by a factor two, and leads to an $f_j$ with large $\|f_j\|_0$, violating Requirement R\ref{eq:fjspare}, as well as R\ref{eq:dfjspare}.

Summarizing, the results in  \Fig~\ref{fig:ilctikresults},  \Fig~\ref{fig:ilctikresults2}, and  \Fig~\ref{fig:spectra} confirm that norm-optimal ILC amplifies trial-varying disturbances, and leads to a non-sparse solution in view of Requirement R\ref{eq:fjspare} and Requirement R\ref{eq:dfjspare}.

\begin{figure}%
\centering
\includegraphics[width=.9\linewidth]{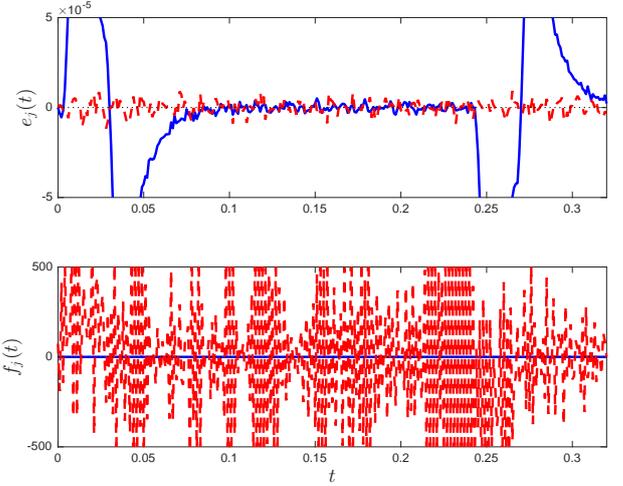}%
\caption{Top: error $e_{j}$. Bottom: command signal $f_{j}$. Shown are iteration $j=0$ (\Colorone), iteration $j=40$  for traditional norm-optimal ILC of \Sec~\ref{sec:exampleNOILC} (\Colortwo).}
\label{fig:ilctikresults}
\end{figure}

\begin{figure}%
\centering
\includegraphics[width=.9\linewidth]{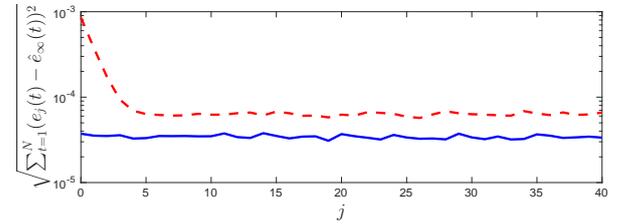}%
\caption{Iteration $j=0$ (\Colorone), iteration $j=40$  for traditional norm-optimal ILC of \Sec~\ref{sec:exampleNOILC} (\Colortwo).}
\label{fig:ilctikresults2}
\end{figure}

\begin{figure}%
\centering
\includegraphics[width=.9\linewidth]{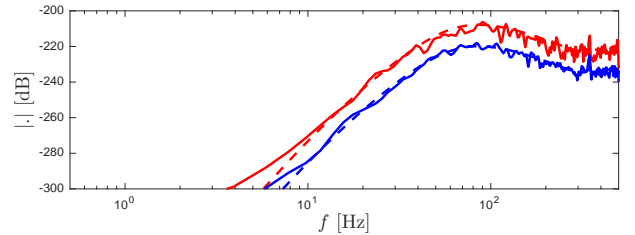}%
\caption{Estimated spectrum of trial-varying part of the error without ILC (solid blue) and for traditional norm-optimal ILC of \Sec~\ref{sec:exampleNOILC} (solid red). Also shown are the spectra $2\phi_v$ (dashed blue) and $\phi_v$ (dashed red).}
\label{fig:spectra}
\end{figure}
 
\subsection{Lasso ILC}\label{sec:exampleLassoILC}

To address Requirement R\ref{eq:fjspare}, the approach in \Sec~\ref{sec:lassoilc} is applied. In particular, $W_e = I$, $W_f = 0$, and $W_{\df} = 0$, $D = I$, and $\lambda =  5 \cdot 10^{-9}$. Next, the ILC iteration is started, and after $40$ iterations it leads to $e_{40}$ and $f_{40}$ in \Fig~\ref{fig:ilclasresults}. Interestingly, $\|f_{40}\|_0$ is much smaller for the lasso ILC approach compared to the results of \Sec~\ref{sec:exampleNOILC}, as is confirmed in \Fig~\ref{fig:ilclasresults3}, thereby addressing Requirement R\ref{eq:fjspare}.  

Also, the $2$-norm of the error signal is computed, see \Fig~\ref{fig:ilclasresults2}. Clearly, the error reduces significantly over the iterations. Finally, also the re-estimated lasso, as is explained in \Sec~\ref{sec:extensions}, is implemented. The results are also depicted in \Fig~\ref{fig:ilclasresults2}. Interestingly, it can be observed that re-estimating leads to a smaller limit error, as is expected. However, note that during the iterations, the approach of \Sec~\ref{sec:lassoilc} leads to a smaller error compared to the re-estimated version in several of the initial iterations. An explanation for this aspect is that the biased estimate leads to a smaller overall error, which is a similar effect as in the Stein estimator. Hence, it is concluded that for non-iterative approaches, the biased estimate can be useful in terms of a bias/variance trade-off, but in the iterative schemes the benefit of re-estimation is clearly confirmed in \Fig~\ref{fig:ilclasresults}.

\begin{figure}%
\centering
\includegraphics[width=.9\linewidth]{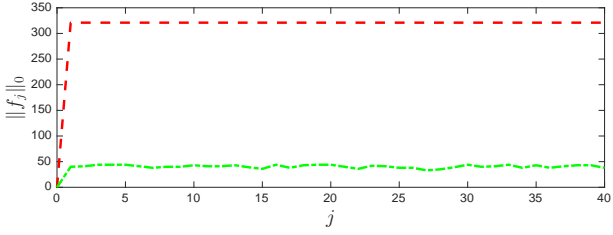}%
\caption{$\ell_0$-norm of the error for norm-optimal ILC of \Sec~\ref{sec:exampleNOILC} (\Colortwo) and lasso ILC of \Sec~\ref{sec:exampleLassoILC} (\Colorthree), which leads to a reduced error signal.}
\label{fig:ilclasresults3}
\end{figure}

\begin{figure}%
\centering
\includegraphics[width=.9\linewidth]{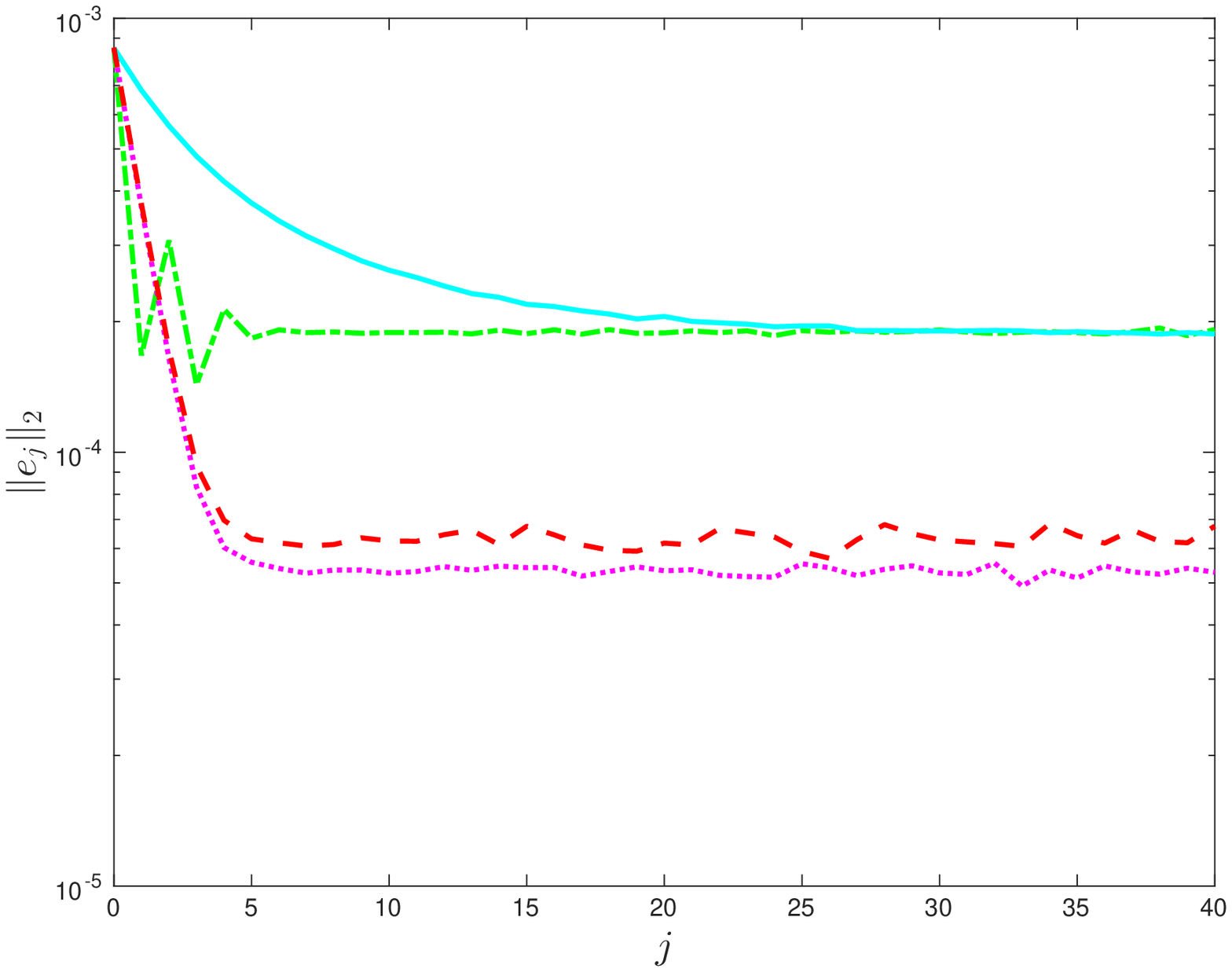}%
\caption{Computed 2-norm of the error for various ILC algorithms. Traditional norm-optimal ILC in \Sec~\ref{sec:exampleNOILC} (\Colortwo) leads to a significant error reduction in the initial iterations, and then remains at a certain level due to the amplification of trial-varying disturbances. The lasso approach of \Sec~\ref{sec:exampleLassoILC} with re-estimation (\Colorfour) leads to a significant reduction in the initial iterations, in addition to a reduced limit error, since it reduces amplification of trial-varying disturbances. Also note that the lasso approach without re-estimation (\Colorthree) leads to an improved estimate in the first iteration, yet remains at a large error after convergence, which is due to the bias error in the solution. Finally, the elastic-net lasso approach of \Sec~\ref{sec:exampleENLassoILC} is shown (\Colorfive), which leads to a comparable converged performance as the lasso ILC (\Colorthree), since both do not include re-estimation in this case.}
\label{fig:ilclasresults2}
\end{figure}

\begin{figure}%
\centering
\includegraphics[width=.9\linewidth]{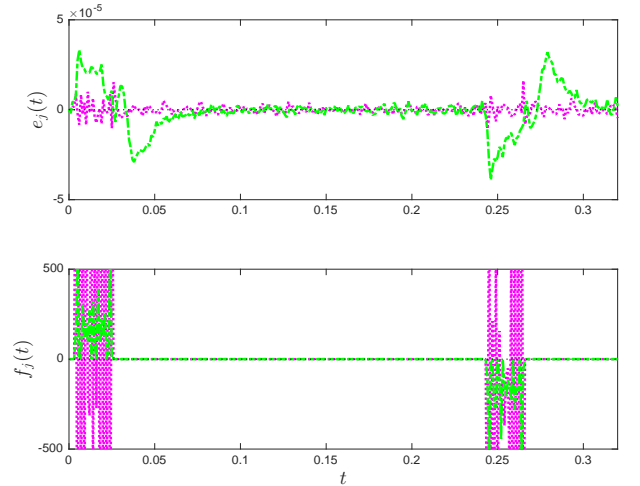}%
\caption{Top: error $e_{40}$ at iteration $j=40$. Bottom: command signal $f_{40}$ at iteration $j=40$. Shown are lasso ILC of \Sec~\ref{sec:exampleLassoILC} (\Colorthree) and re-estimated lasso ILC  of \Sec~\ref{sec:exampleLassoILC} (\Colorfour).}
\label{fig:ilclasresults}
\end{figure}

\subsection{Elastic net lasso ILC}\label{sec:exampleENLassoILC}

In this section, the approach of \Sec~\ref{sec:elasticnet} is pursued, where the lasso regularisation is extended with a ridge regression term. In particular, $W_f = 0$, while $W_{\df} = 1\cdot 10^{-6} I$. The resulting error $e_{40}$ and command input $f_{40}$ are depicted in \Fig~\ref{fig:ilclasENresults}. The error is of comparable magnitude as the lasso ILC in \Sec~\ref{fig:ilclasresults}, while the command input is substantially smoother. The error in fact has slightly reduced compared to lasso ILC, as is shown in \Fig~\ref{fig:ilclasresults2}, which comes at the price of a slower convergence rate due to an increased $W_{\df}$. Notice that the elastic net lasso can also be improved by re-estimation, which is not done here to facilitate the presentation.

\begin{figure}%
\centering
\includegraphics[width=.9\linewidth]{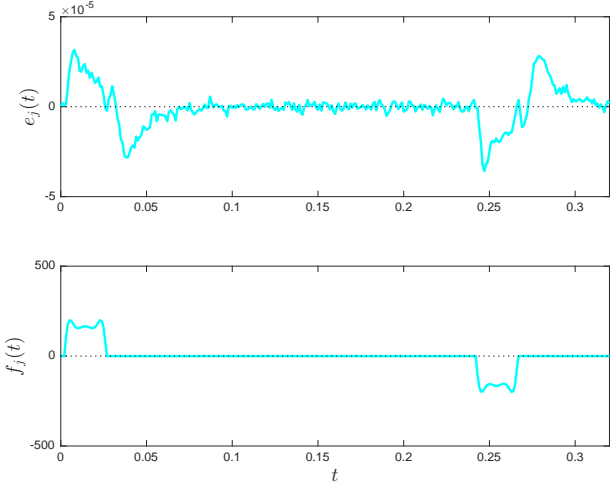}%
\caption{Top: error $e_{40}$ at iteration $j=40$. Bottom: command signal $f_{40}$ at iteration $j=40$. Shown is the elastic net lasso ILC of \Sec~\ref{sec:exampleENLassoILC} (\Colorfive), leading to a smooth command input $f_j$.}
\label{fig:ilclasENresults}
\end{figure}

\subsection{Fused lasso ILC}\label{sec:examplefusedLassoILC}

The results in the previous sections have addressed Requirement R\ref{eq:fjspare}. In certain situations, e.g., wireless sensors or embedded implementations, it may be required to minimize the number of times the command input is updated, i.e., Requirement R\ref{eq:dfjspare}. This is a different form of structure compared to sparsity. To address this, the fused lasso of \Sec~\ref{sec:fusedlasso} is employed. 

In particular, in the general criterion \eqref{eq:gencrit} is considered, where the weighting filters are selected as $W_f = W_\df = 0$, $D = D_f$ in \eqref{eq:fusedlassoD}, and $\lambda = 3 \cdot 10^{-12}$. 

Next, the ILC iteration is invoked. The results are shown in \Fig~\ref{fig:ilclasfusedresults}. Compared to the results of \Fig~\ref{fig:ilclasresults} in \Sec~\ref{sec:exampleLassoILC}, the error has reduced significantly. However, this comes at the price of sparsity. Indeed, only the first samples are zero, since the algorithm is initialized with $f_1(0) = 0$. Interestingly, only a limited number of command signal updates are required to achieve a small error signal. This will also attenuate the effect of trial-varying disturbances. Note that the error can be further reduced by including a re-estimation step, which is not shown here to facilitate the presentation.

\begin{figure}%
\centering
\includegraphics[width=.9\linewidth]{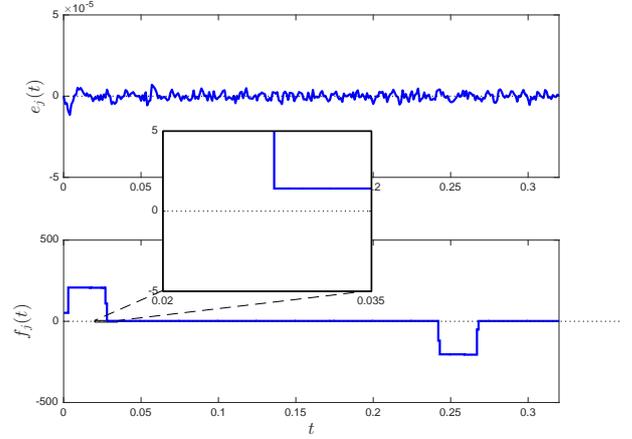}%
\caption{Top: error $e_{40}$ at iteration $j=40$. Bottom: command signal $f_{40}$ at iteration $j=40$. Shown is the fused lasso ILC of \Sec~\ref{sec:examplefusedLassoILC} (\Colorfive), leading to a command input $f_j$ that addresses Requirement R\ref{eq:dfjspare}. In particular, the command signal $f_{40}$ aims to minimize the error signal in addition to the updates, i.e., instants where $f_{40}$ changes as a function of time. This does not explicitly address sparsity of $f_{40}$ itself, as can be clearly observed from the zoom plot.}
\label{fig:ilclasfusedresults}
\end{figure}

\subsection{Sparse fused lasso ILC}\label{sec:examplesparsefusedLassoILC}
In the previous sections, Requirement R\ref{eq:fjspare}  and Requirement R\ref{eq:dfjspare} are achieved separately in \Sec~\ref{sec:exampleLassoILC} and \Sec~\ref{sec:examplefusedLassoILC}, respectively. To address both requirements simultaneously, the sparse fused lasso approach of \Sec~\ref{sec:sparsefusedlasso} is adopted. The regularization penalties in \eqref{eq:sparsefusedlasso} are selected such that these essentially combine the two penalties in \Sec~\ref{sec:exampleLassoILC} and \Sec~\ref{sec:examplefusedLassoILC}.

The results are depicted in \Fig~\ref{fig:ilclassparsefusedresults}. It can directly be observed that it combines the sparsity of \Sec~\ref{sec:exampleLassoILC} while at the same time reducing the number of command signal updates as in \Sec~\ref{sec:examplefusedLassoILC}. As such, it is concluded that the sparse fused lasso addresses Requirement R\ref{eq:fjspare}  and Requirement R\ref{eq:dfjspare} simultaneously. The relative penalties can be further tuned to balance the importance of both penalties, as well as the resulting error signal. In addition, the resulting error signal can be further enhanced through a re-estimation step.

\begin{figure}%
\centering
\includegraphics[width=.9\linewidth]{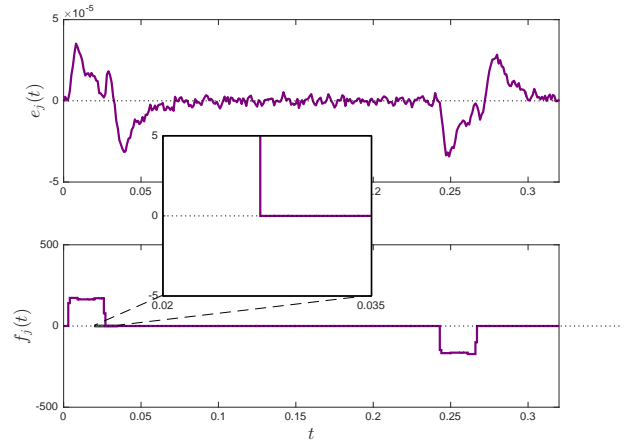}%
\caption{Top: error $e_{40}$ at iteration $j=40$. Bottom: command signal $f_{40}$ at iteration $j=40$. Shown is the sparsefused  lasso ILC of \Sec~\ref{sec:examplesparsefusedLassoILC} (\Colorfive), leading to a command input $f_j$ that addresses Requirement R\ref{eq:dfjspare}. In comparison to the fused lasso approach in Sec~\ref{sec:examplefusedLassoILC}, here additional regularization parameters enforce a zero input signal, which can clearly be seen in the zoom plot.}
\label{fig:ilclassparsefusedresults}
\end{figure}

\section{Conclusion}

A general framework is presented that extends optimization-based iterative learning control to include additional structure, including sparsity. The approach is shown on a mechatronic system, where it is shown to have significant benefits, including 
\begin{inparaenum}[i)]
\item resource-efficiency in terms of sparse command signals, e.g., facilitating embedded controller implementations;
\item resource-efficiency in terms of limiting the number of changes in the command signal, e.g., facilitating implementation in limited-capacity communication networks;
\item automated basis function selection in flexible iterative learning control employing basis functions; and
\item attenuation of trial-varying disturbances, which for the considered wafer scanner example leads to significant performance increase.
\end{inparaenum}
Regarding the latter, a detailed analysis of trial-varying disturbances in ILC reveals that such trial-varying exogenous signals are often amplified by typical ILC algorithms. The proposed framework enables a significant reduction of this amplification, typically up to a factor of two. 

The proposed framework enables many user-specific choices, and can be easily extended. For instance, for third-order or higher-order setpoints, it may be useful to impose regularization parameters of equally high polynomial orders, known as polynomial trend filtering \cite[\Sec\ 2.1.2]{TibshiraniTay2011}, which is a special case of the general criterion \eqref{eq:gencrit}. 

Ongoing research focusses on specialized algorithms for the considered scenarios, enabling faster computation. In addition, the correlation between variables is subject of further investigation. Finally, various aspects of monotonic convergence, which has here been analyzed in terms of the $\ell_1$-norm, are being investigated, including robust monotonic convergence conditions \cite{WijdevenDonBos2009}, \cite{DuysonPipSwe2016} and data-driven ILC frameworks \cite{JanssensPipSwe2013}, \cite{AAABolderKleOom}.%

\section*{Appendix}
In this section, a proof of Theorem~\ref{thm:noiseanalysis} is provided. Several auxiliary results are presented. In particular, note that at iteration $j$, the error is a function of all previous signals affecting the loop due to the iteration-domain integrator in \eqref{eq:freqilcupdate}. In the following lemma, the summation of $j$ terms of the trial-invariant disturbance $r$ in \eqref{eq:generalILCsystem} is eliminated.

\begin{lemma}\label{lemma:firststepproof}
Consider the system \eqref{eq:generalILCsystem} and ILC update \eqref{eq:freqilcupdate} with $f_0 = 0$ and assume that iteration is stable in the sense of Theorem~\ref{thm:contractionmap}. Then,
\begin{align}\label{eq:eliminater}
e_j =& \left(
1 - \tf[J] \frac{1-(\tf[Q](1-\tf[L]\tf[J]))^{j}}{1-\tf[Q](1-\tf[L]\tf[J])}\tf[Q]\tf[L]\right)r\\
&-
v_j
-
\tf[J] \sum_{n=0}^{j-1}(\tf[Q](1-\tf[L]\tf[J])^n\tf[Q]\tf[L]v_{j-n-1}.
\end{align}
\end{lemma}

\begin{proof}
Substituting \eqref{eq:generalILCsystem} into \eqref{eq:freqilcupdate} yields
\begin{equation}
f_{j+1} = \tf[Q]((1-\tf[L]\tf[J])f_j + \tf[L]r - \tf[L] v_j).
\end{equation}
Given $f_0 = 0$ and subsequent successive substitution yields
\begin{math}
f_1 = \tf[Q](\tf[L]r - \tf[L]v_0)
\end{math}, 
\begin{math}
f_2 = \tf[Q]((1-\tf[L]\tf[J])+1)\tf[Q]\tf[L]r - \tf[Q](1-\tf[L]\tf[J])\tf[Q]\tf[L]v_0 -\tf[Q]\tf[L]v_1,
\end{math}
and hence
\begin{equation}
f_j = \sum_{i=0}^{j-1}(\tf[Q](1-\tf[L]\tf[J]))^{i} \tf[Q]\tf[L]r - \sum_{n=0}^{j-1}(\tf[Q](1-\tf[L]\tf[J]))^{n}\tf[Q]\tf[L]v_{j-1-n}.
\end{equation}
Next, using the geometric series 
\begin{equation}\label{eq:geometricseries}
\sum_{l = 0}^{j-1} r^l = \frac{1 - r^j}{1-r},
\end{equation}
this leads to
\begin{equation}\label{eq:almostfinalstep}
f_j = \frac{1-(\tf[Q](1-\tf[L]\tf[J]))^{j}}{1-\tf[Q](1-\tf[L]\tf[J])}\tf[Q]\tf[L]r - \sum_{n=0}^{j-1}(\tf[Q](1-\tf[L]\tf[J]))^{n}\tf[Q]\tf[L]v_{j-1-n}.
\end{equation}
Finally, substitution of \eqref{eq:almostfinalstep} into \eqref{eq:generalILCsystem} yields the desired result \eqref{eq:eliminater}.
\end{proof}

The result~\eqref{eq:eliminater} reveals that the error contains a summation over $j$ trial-varying disturbance terms $v_j$, whereas the influence of the trial-invariant disturbances is captured in a single term through the use of a geometric series. Although the trial-varying disturbance varies on each experiment, a closed-form expression can be obtained by exploiting Assumption~\ref{assum:noise}. 

\begin{lemma}\label{lemma:secondstepproof}
Let Assumption~\ref{assum:noise} hold. Then, under the assumptions of Lemma~\ref{lemma:firststepproof}, 
\begin{align}\label{eq:eliminatev}
\phi_{e_j} =& \left| 
1 - \tf[J] \frac{1-(\tf[Q](1-\tf[L]\tf[J]))^{j}}{1-\tf[Q](1-\tf[L]\tf[J])}\tf[Q]\tf[L]
\right|^2 \phi_r \\ 
&+
\left(
1
+
\left|
\tf[J]\tf[Q]\tf[L]
\right|^2
\frac{1-|\tf[Q](1-\tf[L]\tf[J])|^{2j}}{1-|\tf[Q](1-\tf[L]\tf[J])|^2}
\right)\phi_v
\end{align}
\end{lemma}
\begin{proof}
Taking spectra yields
\begin{align}
\phi_{e_j} =& \left| 
1 - \tf[J] \frac{1-(\tf[Q](1-\tf[L]\tf[J]))^{j}}{1-\tf[Q](1-\tf[L]\tf[J])}\tf[Q]\tf[L]
\right|^2 \phi_r \\ 
&+
\left(
1
+
\left|
\tf[J]\tf[Q]\tf[L]
\right|^2
\sum_{n=0}^{j-1}
\left|
(\tf[Q](1-\tf[L]\tf[J]))^n
\right|^2
\right)\phi_v.
\end{align}
Next, using \eqref{eq:geometricseries} yields the desired result \eqref{eq:eliminatev}. 
\end{proof}

The closed-form solution~\eqref{eq:eliminatev} enables a direct proof of Theorem~\ref{thm:noiseanalysis}.

\begin{proof}(of Theorem~\ref{thm:noiseanalysis})
Taking the limit $j \rightarrow \infty$ implies that $|(1-\tf[L]\tf[J]))^{j}| \rightarrow 0$, directly leading to the desired result \eqref{eq:limiterrorspectrum}.
\end{proof}

\section*{Acknowledgements}

This paper is the result of several research visits of both authors, which is supported in part of the research programme VENI with project number 13073, which is (partly) financed by the Netherlands Organisation for Scientific Research (NWO). In addition, the authors gratefully acknowledge the fruitful discussions with Jurgen van Zundert, Maurice Heemels, Dip Goswami, and Martijn Koedam for resource-efficient control, as part of the  Robust Cyber-Physical Systems (RCPS) project (no. 12694).

\bibliographystyle{abbrv}

\begin{thebibliography}{10}

\bibitem{AhnMooChe2007}
H.-S. Ahn, K.~L. Moore, and Y.~Chen.
\newblock {\em Iterative Learning Control: Robustness and Monotonic Convergence
  for Interval Systems}.
\newblock Communications and Control Engineering Series. Springer-Verlag,
  London, \UK, 2007.

\bibitem{AnnergrenHanWah2012}
M.~Annergren, A.~Hansson, and B.~Wahlberg.
\newblock An {ADMM} algorithm for solving $\ell_1$ regularized {MPC}.
\newblock In {\em \CDC[2012]}, pages 4486--4491, Maui, \HI, 2012.

\bibitem{BachJenMaiObo2011}
F.~Bach, R.~Jenatton, J.~Mairal, and G.~Obozinski.
\newblock {\em Optimization with Sparsity-Inducing Penalties}, volume~4 of {\em
  Foundations and Trends in Machine Learning}.
\newblock 2011.

\bibitem{BertsimasKinMaz2016}
D.~Bertsimas, A.~King, and R.~Mazumder.
\newblock Best subset selection via a modern optimization lens.
\newblock {\em The Annals of Statistics}, 44(2), 2016.

\bibitem{BoerenBarKokOom2016}
F.~Boeren, A.~Bareja, T.~Kok, and T.~Oomen.
\newblock Frequency-domain {ILC} approach for repeating and varying tasks: With
  application to semiconductor bonding equipment.
\newblock {\em \IEEETM}, 21(6):2716--2727, 2016.

\bibitem{BoerenBruOom2017}
F.~Boeren, D.~Bruijnen, and T.~Oomen.
\newblock Enhancing feedforward controller tuning via instrumental variables:
  With application to nanopositioning.
\newblock {\em \IJC}, 90(4):746--764, 2017.

\bibitem{AAABolderKleOom}
J.~Bolder, S.~Kleinendorst, and T.~Oomen.
\newblock Data-driven multivariable {ILC}: Enhanced performance by eliminating
  ${L}$ and ${Q}$ filters.
\newblock {\em \IJRNC}, To appear.

\bibitem{BolderOomKoeSte2014c}
J.~Bolder, T.~Oomen, S.~Koekebakker, and M.~Steinbuch.
\newblock Using iterative learning control with basis functions to compensate
  medium deformation in a wide-format inkjet printer.
\newblock {\em \MECH}, 24(8):944--953, 2014.

\bibitem{BolderZunKoeOom2017}
J.~Bolder, J.~van Zundert, S.~Koekebakker, and T.~Oomen.
\newblock Enhancing flatbed printer accuracy and throughput: Optimal rational
  feedforward controller tuning via iterative learning control.
\newblock {\em \IEEETIE}, 64(5):4207--4216, 2017.

\bibitem{Breiman1995}
L.~Breiman.
\newblock Better subset regression using the nonnegative garrote.
\newblock {\em Technometrics}, 37(4):373--384, 1995.

\bibitem{Bristow2008}
D.~A. Bristow.
\newblock Weighting matrix design for robust monotonic convergence in norm
  optimal iterative learning control.
\newblock In {\em \ACC[2008]}, pages 4554--4560, Seattle, \WA, 2008.

\bibitem{BristowThaAll2006}
D.~A. Bristow, M.~Tharayil, and A.~G. Alleyne.
\newblock A survey of iterative learning control: A learning-based method for
  high-performance tracking control.
\newblock {\em \CSM}, 26(3):96--114, 2006.

\bibitem{BuhlmannGee2011}
P.~B\"uhlmann and S.~van~de Geer.
\newblock {\em Statistics for High-Dimensional Data}.
\newblock Springer Series in Statistics. Springer, Heidelberg, Germany, 2011.

\bibitem{ButcherKar2011}
M.~Butcher and A.~Karimi.
\newblock {\em Advances in the Theory of Control, Signals and Systems with
  Physical Modeling}, chapter Iterative Learning Control Using Stochastic
  Approximation Theory with Application to a Mechatronic System, pages 49--64.
\newblock Number 407 in \LNCIS. Springer, 2011.

\bibitem{Butler2011}
H.~Butler.
\newblock Position control in lithographic equipment an enabler for current-day
  chip manufacturing.
\newblock {\em \CSM}, 31(5):28--47, 2011.

\bibitem{CandesWakBoy2008}
E.~Cand\`es, M.~B. Wakin, and S.~P. Boyd.
\newblock Enhancing sparsity by reweighted $\ell_1$ minimization.
\newblock {\em J. Fourier Anal. Appl.}, 14:877--905, 2008.

\bibitem{CandesTao2005}
E.~J. Cand\`es and T.~Tao.
\newblock Decoding by linear programming.
\newblock {\em {IEEE} Transactions on Information Theory}, 51(12):4203--4215,
  2005.

\bibitem{DuysonPipSwe2016}
T.~Duy~Son, G.~Pipeleers, and J.~Swevers.
\newblock Robust monotonic convergent iterative learning control.
\newblock {\em \IEEETAC}, 61(4):1063--1068, 2016.

\bibitem{FeliciOom2015}
F.~Felici and T.~Oomen.
\newblock Enhancing current density profile control in tokamak experiments
  using iterative learning control.
\newblock In {\em \CDC[54]}, pages 5370--5377, Osaka, Japan, 2015.

\bibitem{FreemanHugBurChaLewRog2009}
C.~Freeman, A.-M. Hughes, J.~Burridge, P.~Chappell, P.~Lewin, and E.~Rogers.
\newblock Iterative learning control of {FES} applied to the upper extremity
  for rehabilitation.
\newblock {\em \CEP}, 17(3):368--381, 2009.

\bibitem{Gallieri2015}
M.~Gallieri.
\newblock {\em $\ell_{asso}$-{MPC} - Predictive Control with
  $\ell_1$-Regularised Least Squares}.
\newblock Springer Theses. Springer, Switzerland, 2015.

\bibitem{GoossensAzeChaDevGooKoeLiMirMolBeyNelSin2013}
K.~Goossens, A.~Azevedo, K.~Chandrasekar, M.~Dev~Gomony, S.~Goossens,
  M.~Koedam, Y.~Li, D.~Mirzoyan, A.~Molnos, A.~Beyranvand~Nejad, A.~Nelson, and
  S.~Sinha.
\newblock Virtual execution platforms for mixed-time-criticality systems: The
  {CompSOC} architecture and design flow.
\newblock {\em {SIGBED} Review}, 10(3):23--34, 2014.

\bibitem{GunnarssonNor2001}
S.~Gunnarsson and M.~Norrl{\"o}f.
\newblock On the design of {ILC} algorithms using optimization.
\newblock {\em \AUT}, 37:2011--2016, 2001.

\bibitem{GunnarssonNor2006}
S.~Gunnarsson and M.~Norrl\"of.
\newblock On the disturbance properties of high order iterative learning
  control algorithms.
\newblock {\em \AUT}, 42:2031--2034, 2006.

\bibitem{HastieTibWai2015}
T.~Hastie, R.~Tibshirani, and M.~Wainwright.
\newblock {\em Statistical Learning with Sparsity: The Lasso and
  Generalizations}.
\newblock CRC Press, 2015.

\bibitem{HoelzleAllWag2011}
D.~J. Hoelzle, A.~G. Alleyne, and A.~J. Wagoner~Johnson.
\newblock Basis task approach to iterative learning control with applications
  to micro-robotic deposition.
\newblock {\em \IEEECST}, 19(5):1138--1148, 2011.

\bibitem{HoelzleBar2016}
D.~J. Hoelzle and K.~L. Barton.
\newblock On spatial iterative learning control via {2-D} convolution:
  Stability analysis and computational efficiency.
\newblock {\em \IEEECST}, 24(4):1504--1512, 2016.

\bibitem{JanssensPipSwe2013}
P.~Janssens, G.~Pipeleers, and J.~Swevers.
\newblock A data-driven constrained norm-optimal iterative learning control
  framework for {LTI} systems.
\newblock {\em \IEEECST}, 21(2):546--551, 2013.

\bibitem{KhoshfetratOhlLju2013}
S.~Khoshfetrat~Pakazad, H.~Ohlsson, and L.~Ljung.
\newblock Sparse control using sum-of-norms regularized model predictive
  control.
\newblock In {\em \CDC[2013]}, pages 5758--5763, Firenze, Italy, 2013.

\bibitem{KimKohBoyGor2009}
S.-J. Kim, W.~Koh, S.~Boyd, and D.~Gorinevsky.
\newblock $\ell_1$ trend filtering.
\newblock {\em {SIAM} Review}, 51(2):339--360, 2009.

\bibitem{LambrechtsBoeSte2005}
P.~Lambrechts, M.~Boerlage, and M.~Steinbuch.
\newblock Trajectory planning and feedforward design for electromechanical
  motion systems.
\newblock {\em \CEP}, 13:145--157, 2005.

\bibitem{Ljung1999}
L.~Ljung.
\newblock {\em System Identification: Theory for the User}.
\newblock Prentice Hall, Upper Saddle River, \NJ, second edition, 1999.

\bibitem{MerryMolSte2008}
R.~Merry, R.~van~de Molengraft, and M.~Steinbuch.
\newblock Iterative learning control with wavelet filtering.
\newblock {\em \IJRNC}, 18(10):1052--1071, 2008.

\bibitem{MeulenTouBos2008}
S.~{\SortAt{Meulen}}van~der Meulen, R.~L. Tousain, and O.~H. Bosgra.
\newblock Fixed structure feedforward controller design exploiting iterative
  trials: Application to a wafer stage and a desktop printer.
\newblock {\em \JDMC}, 130:051006--1, 2008.

\bibitem{Moore1993}
K.~L. Moore.
\newblock {\em Iterative Learning Control for Deterministic Systems}.
\newblock Springer-Verlag, London, \UK, 1993.

\bibitem{Murphy2012}
K.~P. Murphy.
\newblock {\em Machine Learning: A Probabilistic Perspective}.
\newblock The {MIT} Press, 2012.

\bibitem{Natarajan1995}
B.~K. Natarajan.
\newblock Sparse approximate solutions to linear systems.
\newblock {\em {SIAM} J. Comput.}, 24(2):227--234, 1995.

\bibitem{OhlssonLjuBoy2010}
H.~Ohlsson, L.~Ljung, and S.~Boyd.
\newblock Segmentation of {ARX}-models using sum-of-norms regularization.
\newblock {\em \AUT}, 46(6):1107--1111, 2010.

\bibitem{Oomen2014}
T.~Oomen.
\newblock Controlling aliased dynamics in motion systems? {An} identification
  for sampled-data control approach.
\newblock {\em \IJC}, 87(7):1406--1422, 2014.

\bibitem{OomenHerQuiWalBosSte2014}
T.~Oomen, R.~van Herpen, S.~Quist, M.~van~de Wal, O.~Bosgra, and M.~Steinbuch.
\newblock Connecting system identification and robust control for
  next-generation motion control of a wafer stage.
\newblock {\em \IEEECST}, 22(1):102--118, 2014.

\bibitem{Owens2016}
D.~H. Owens.
\newblock {\em Iterative Learning Control: An Optimization Paradigm}.
\newblock Advances in Industrial Control. Springer, 2016.

\bibitem{PengSunZhaTom2016}
C.~Peng, L.~Sun, W.~Zhang, and M.~Tomizuka.
\newblock Optimization-based constrained iterative learning control with
  application to building temperature control systems.
\newblock In {\em \AIM[2016]}, pages 709--715, Banff, Alberta, Canada, 2016.

\bibitem{PipeleersMoo2014}
G.~Pipeleers and K.~L. Moore.
\newblock Unified analysis of iterative learning and repetitive controllers in
  trial domain.
\newblock {\em \IEEETAC}, 59(4):953--965, 2014.

\bibitem{RogersGalOwe2007}
E.~Rogers, K.~Galkowski, and D.~H. Owens.
\newblock {\em Control Systems Theory and Applications for Linear Repetitive
  Processes}.
\newblock Number 349 in \LNCIS. Springer, Berlin, Germany, 2007.

\bibitem{RojasHja2011}
C.~R. Rojas and H.~Hjalmarsson.
\newblock Sparse estimation based on a validation criterion.
\newblock In {\em \CDC[2011]}, pages 2825--2830, Orlando, \FL, 2011.

\bibitem{RojasTotHja2014}
C.~R. Rojas, R.~T\'oth, and H.~Hjalmarsson.
\newblock Sparse estimation of polynomial and rational dynamical models.
\newblock {\em \IEEETAC}, 59(11):2962--2977, 2014.

\bibitem{Stein1956}
C.~Stein.
\newblock Inadmissibility of the usual estimator for the mean of a multivariate
  normal distribution.
\newblock In {\em Proceedings of the Third Berkeley Symposium on Mathematical
  Statistics and Probability}, volume~1, pages 197--206. Univ. of Calif. Pres,
  1956.

\bibitem{Tibshirani1996}
R.~Tibshirani.
\newblock Regression shrinkage and selection via the lasso.
\newblock {\em Journal of the Royal Statistical Society}, 58(1):267--288, 1996.

\bibitem{TibshiraniSauRosZhuKni2005}
R.~Tibshirani, M.~Saunders, S.~Rosset, J.~Zhu, and K.~Knight.
\newblock Sparsity and smoothness via the fused lasso.
\newblock {\em J. R. Statist. Soc. B}, 67(1):91--108, 2005.

\bibitem{TibshiraniTay2011}
R.~J. Tibshirani and J.~Taylor.
\newblock The solution path of the generalized lasso.
\newblock {\em The Annals of Statistics}, 39(3):1335--1371, 2011.

\bibitem{WallenDreGunRob2014}
J.~Wall\'en~Axehill, I.~Dressler, S.~Gunnarsson, and A.~Robertsson.
\newblock Estimation-based {ILC} applied to a parallel kinematic robot.
\newblock {\em \CEP}, 33:1--9, 2014.

\bibitem{WijdevenBos2010}
J.~{\SortAt{Wijdeven}}van~de Wijdeven and O.~Bosgra.
\newblock Using basis functions in iterative learning control: Analysis and
  design theory.
\newblock {\em \IJC}, 83(4):661--675, 2010.

\bibitem{WijdevenDonBos2009}
J.~{\SortAt{Wijdeven}}van~de Wijdeven, T.~Donkers, and O.~Bosgra.
\newblock Iterative learning control for uncertain systems: Robust monotonic
  convergence analysis.
\newblock {\em \AUT}, 45:2383--2391, 2009.

\bibitem{YuanLin2006}
M.~Yuan and Y.~Lin.
\newblock Model selection and estimation in regression with grouped variables.
\newblock {\em Journal of the Royal Statistical Society, Series B}, 68:49--67.

\bibitem{ZhouDoyGlo1996}
K.~Zhou, J.~C. Doyle, and K.~Glover.
\newblock {\em Robust and Optimal Control}.
\newblock Prentice Hall, Upper Saddle River, \NJ, 1996.

\bibitem{Zou2006}
H.~Zou.
\newblock The adaptive lasso and its oracle properties.
\newblock {\em Journal of the American Statistical Association},
  101(476):1418--1429, 2006.

\bibitem{ZouHas2005}
H.~Zou and T.~Hastie.
\newblock Regularization and variable selection via the elastic net.
\newblock {\em J. R. Statist. Soc. B}, 67(2):301--320, 2005.

\bibitem{ZundertBolKoeOom2016b}
J.~{\SortAt{Zundert}}van~Zundert, J.~Bolder, S.~Koekebakker, and T.~Oomen.
\newblock Resource-efficient {ILC} for {LTI/LTV} systems through {LQ} tracking
  and stable inversion: Enabling large tasks on a position-dependent industrial
  printer.
\newblock {\em \MECH}, 38:76--90, 2016.

\bibitem{ZundertBolOom2015}
J.~{\SortAt{Zundert}}van~Zundert, J.~Bolder, and T.~Oomen.
\newblock Iterative learning control for varying tasks: Achieving optimality
  for rational basis.
\newblock In {\em \ACC[2015]}, Chicago, \IL, 2015.

\end{thebibliography}

\end{document}